\documentclass[12pt, leqno]{article}

\usepackage{amsthm} \usepackage{amsmath} \usepackage{amsfonts}
\usepackage{amssymb} \usepackage{latexsym} \usepackage{enumerate}
\usepackage[dvips]{color}
 \usepackage{mathrsfs}
\usepackage[numbers]{natbib}

\theoremstyle{plain}%
\newtheorem{Theorem}{Theorem}[section] %
\newtheorem{Lemma}[Theorem]{Lemma}
\newtheorem{Proposition}[Theorem]{Proposition} %

\theoremstyle{definition}%
\newtheorem{Assumption}[Theorem]{Assumption}%

\theoremstyle{remark}%
\newtheorem{Remark}[Theorem]{Remark} %

\newcommand{\set}{\triangleq}

\newcommand{\envspace}{\vspace{2mm}}

\renewcommand{\mathcal}{\mathscr}

\renewcommand{\epsilon}{\varepsilon}

\numberwithin{equation}{section}

\begin{document}

\title{Optimal
  investment with intermediate consumption and random endowment}
\author{Oleksii Mostovyi\\
  University of Texas at Austin,\\
  Department of  Mathematics,\\
  2515 Speedway Stop C1200,\\
  Austin, Texas 78712-1202\\
  (mostovyi@math.utexas.edu)}
\date{\today}

\date{}

\maketitle
\begin{abstract}
We consider a problem of optimal
investment with intermediate consumption and random endowment in
an incomplete semimartingale model of a financial market. We establish the key assertions of the utility
maximization theory assuming that both primal and dual value
functions are finite in the interiors of their domains as well as
that random endowment at maturity can be dominated by the terminal
value of a self-financing wealth process. In order to facilitate
verification of these conditions, we present alternative, but
equivalent conditions, under which the conclusions of the theory
hold.
 \end{abstract}

\section{Introduction}
The problem of utility maximization in incomplete markets is
of central importance in mathematical finance. The theory was
developed by He and Pearson \cite{HePearson1, HePearson2},
Karatzas, Lehoczky,  Shreve, and Xu \cite{KLSX}, Karatzas and Shreve
\cite{KSmmf}, Kramkov and Schachermayer \cite{KS, KS2003},
Karatzas and \v{Z}itkovi\'c \cite{Karatzas-Zitkovic-2003}, and
\v{Z}itkovi\'c \cite{Zitkovic}.

In this paper we consider a problem for an agent, who in addition to
the initial wealth receives random endowment. The goal of such an
agent is to consume and invest in a way that maximizes his expected
utility. In complete market settings this problem is considered by
Karatzas and Shreve \cite{KSmmf}, see Chapter 4. Using replication
argument, the authors were able to reduce the problem to one without
endowment. Since in incomplete markets such replication might not be
possible, alternative techniques were used. For example, Cuoco
\cite{Cuoco1997} used martingale techniques to reformulate the
dynamic optimization problem as an equivalent static one. In
Markovian settings a possible approach is to use a
Hamilton-Jacobi-Bellman equation for the value function, see Duffie
and Zariphopoulou \cite{DuffieZarip1993} and Duffie, Fleming, Soner,
and Zariphopoulou \cite{DuffieFlemingSonerZarip1997}. Cvitani\'c,
Schachermayer, and Wang \cite{CvitanicSchachermayerWang} considered
the problem of optimal investment from terminal wealth under the
presence of random endowment in an incomplete semimartingale market.
Using the space $\left(\mathbf L^{\infty}\right)^{*}$ of finitely
additive measures as the domain of the dual problem they were able
to characterize the value function and the optimal terminal wealth
in terms of the solution to the dual problem.

In contrast to \cite{CvitanicSchachermayerWang}, Hugonnier and
Kramkov \cite{HugonnierKramkov2004} treated not only the initial
capital  as the variable of the value function, but also the
number of shares of random endowment. Although this lead to higher
dimensionality of the problem, such an approach allowed to relax some technical
assumptions. Stability of this utility maximization problem was
investigated by Kardaras and \v{Z}itkovi\'{c}
\cite{KardarasZitkovic2011}. Karatzas and \v{Z}itkovi\'{c}
\cite{Karatzas-Zitkovic-2003} as well as \v{Z}itkovi\'{c}
\cite{Zitkovic} extended the results of Cvitani\'{c}, Schachermayer,
and Wang \cite{CvitanicSchachermayerWang} to include
intermediate consumption.
Mostovyi \cite{Mostovyi2011} considered the problem of optimal
investment with intermediate consumption under the condition that
both primal and dual value functions are finite in their domains. It
is shown in \cite{Mostovyi2011} that such conditions are both
necessary and sufficient for the validity of the ``key'' conclusions of
the theory.

Present work extends the results of Mostovyi \cite{Mostovyi2011} to
incorporate random endowment process into the model. As in
\cite{HugonnierKramkov2004}, we treat the number of shares of random
endowment as the variable of the value function. This
approach allows us to obtain the standard conclusions of the utility
maximization theory under the following conditions:
(as in \cite{Mostovyi2011}) we assume that  both primal and dual functions are
finite in the interiors of
their domains and (as in \cite{HugonnierKramkov2004}) that the random endowment at maturity can be dominated by
the terminal value of a self-financing nonnegative wealth process. In order to facilitate verification of the former
condition, we present an alternative, but equivalent criterion in
terms of the finiteness of the value functions without random
endowment. The condition on the endowment can also be formulated in
several equivalent ways, which we specify as well.

In addition to the usual conclusions of the utility maximization theory, it turns out to be possible to establish some properties of the
primal and dual value functions on the boundaries of their domains,
namely to show upper semi-continuity of the primal value function
and lower semi-continuity of dual value function.

The remainder of this paper is organized as follows. In Section
\ref{mainResults} we describe the mathematical model and state  our
main results, whose proofs are presented in Section \ref{proofOfMainThm}.

\section{Main Results}\label{mainResults}
We consider a model of a financial market with finite time horizon
$[0,T]$ and zero interest rate. The price process $S=(S^i)_{i=1}^d$
of the stocks is assumed to be a 
semimartingale on a complete stochastic basis
$\left( \Omega, \mathcal F, \left(\mathcal F_t \right)_{t\in[0,T]}, \mathbb
P\right)$, where $\mathcal F_0$ is the completion of the trivial $\sigma$-algebra. We assume that
there are non-traded contingent claims with a payment
process $F = (F^i)_{i=1}^N$. If $q=(q_i)_{i=1}^N$ is the number of
such claims, then the cumulative payoff of this portfolio is given by
\begin{equation}\nonumber
qF \set \sum\limits_{i=1}^N q_iF^i = \left(\sum\limits_{i=1}^N
q_iF^i_t\right)_{t\in[0,T]}.
\end{equation}
Thus, the random variable $qF_t$ stands for the cumulative amount of
endowment received by a holder of $q$ such claims during the time
interval $[0,t].$ Both processes $S$ and $F$ are given exogenously.

As in \cite{Mostovyi2011}, we define a \textit{stochastic clock} as a
nondecreasing, c\'adl\'ag, adapted process such that
\begin{equation}
  \label{stochasticClock}
\kappa_0 = 0, ~~ \mathbb P\left[\kappa_T>0 \right]>0,\text{ and } \kappa_T\leq A
\end{equation}
for some finite constant $A$.  

Define a {portfolio} $\Pi$ as a quadruple $(x, q, H, c),$ where the
constant $x$ is the initial value of the portfolio, vector $q$ gives
the number of shares of illiquid contingent claims, $H =
(H_i)_{i=1}^d$ is a predictable $S$-integrable process that
specifies the amount of each stock in the portfolio, and
$c=(c_t)_{t\in[0,T]}$ is the consumption rate, which we assume to be
 optional and nonnegative.

The \textit{wealth process} $V=(V_t)_{t\in[0,T]}$ of such a
portfolio is defined as
\begin{equation}\nonumber
V_t = x + \int_0^t H_sdS_s - \int_0^tc_sd\kappa_s +
qF_t,\quad t\in[0,T].
\end{equation}
A portfolio $\Pi$ with $c\equiv 0$ and $q = 0$ is called
\textit{self-financing}. The collection of nonnegative wealth
processes of self-financing portfolios with initial value $x\geq 0$ is
denoted by $\mathcal X(x)$, i.e.
\begin{displaymath}
\mathcal X(x) \set \left\{X\geq 0:~X_t=x + \int_0^t
H_sdS_s,~~t\in[0,T] \right\},~~x\geq 0.
\end{displaymath}
A probability measure $\mathbb Q$ is an \textit{equivalent local
martingale measure} if $\mathbb Q$ is equivalent to $\mathbb P$ and
every $X\in\mathcal X(1)$ is a local martingale under $\mathbb Q$. We
denote the family of equivalent local martingale measures by
$\mathcal M$ and assume that
\begin{equation}\label{MisNotEmpty}
\mathcal M \neq \emptyset.
\end{equation}
This condition is essentially equivalent to the absence of arbitrage
opportunities 
on the market
, see
Delbaen and Schachermayer \cite{DS, Delbaen-Schachermayer1998} as well as Karatzas and Kardaras
\cite{Karatzas-Kardaras} for the exact statements and further
references.

To rule out doubling strategies in the presence of random endowment
we need to impose additional restrictions. Following Delbaen and
Schachermayer \cite{Delbaen-Schachermayer1997}, we say that a
nonnegative process in $\mathcal X(x)$ is \textit{maximal} if its
terminal value cannot be dominated by that of any other process in
$\mathcal X(x)$.

As in \cite{Delbaen-Schachermayer1997}, we define an
\textit{acceptable} process to be a process of the form $X = X' -
X'',$ where $X'$ is a nonnegative wealth process of a self-financing portfolio and $X''$ is maximal. Following
Hugonnier and Kramkov \cite{HugonnierKramkov2004}, we denote by
$\mathcal X(x,q)$ the set of acceptable processes with initial value
$x$ whose terminal value dominates the random payoff $-qF_T$:
\begin{equation}\label{defOfXcal}
\mathcal X(x, q) \set
\left\{X:~X{\rm~is~acceptable,~}X_0=x{\rm~and~}X_T+qF_T\geq
0\right\}.
\end{equation}
The set $\mathcal X(x,q)$ may be empty for some $(x,q)\in\mathbb
R^{N+1}$. We are interested in the values of $x$ and $q$,
for which $\mathcal X(x,q)\neq \emptyset$, and define
\begin{equation}\label{defOfK}
\mathcal K \set {\rm int}\left\{ (x,q)\in\mathbb R^{N+1}:~\mathcal
X(x,q)\neq \emptyset \right\}.
\end{equation}
It is proved Lemma 6 in Hugonnier and Kramkov
\cite{HugonnierKramkov2004} that
\begin{equation}\nonumber
{\rm cl}\mathcal K = \left\{ (x,q)\in\mathbb R^{N+1}:~\mathcal
X(x,q)\neq \emptyset \right\},
\end{equation}
where ${\rm cl}\mathcal K$ denotes the closure of the set $\mathcal
K$ in $\mathbb R^{N+1}$.

 We
restrict our attention to the wealth processes with nonnegative
\textit{terminal} values. Thus for each $(x, q)\in{\rm cl}\mathcal
K$ we set 
\begin{equation}\label{defOfA}
\begin{array}{c}
\mathcal A(x, q) \set \left\{ c=(c_t)_{t\in[0, T]} :~c~{\rm
is~nonnegative, ~optional,}\right.\\
\hspace{5mm}\left.{\rm and~there~exists~}X\in\mathcal
X(x,q){\rm~s.t.~}X_T-\int_0^Tc_td\kappa_t+qF_T\geq 0 \right\}.
\end{array}
\end{equation}
Hereafter we shall impose the following conditions on the the endowment process.
\begin{Assumption}\label{assumptionOnEndowment}
$(F^i_T)_{i=1,\dots ,N}$ are $\mathcal F_T$-measurable functions. There exists a maximal
nonnegative wealth process $X'$ of a
self-financing portfolio, such that
\begin{equation}\label{conditonOnEndowment}
X'_T\geq \sum\limits_{i=1}^N |F^i_T|.
\end{equation}
\end{Assumption}
\begin{Remark}\label{remarkOnTheEndowmnet}
Since in the definition (\ref{defOfA})  the endowment process $F$
enters only via its terminal value, it is natural to impose a
regularity condition on $F_T$ (and not on the whole $F$) as in the
Assumption \ref{assumptionOnEndowment}. 
Under the assumptions of the Theorem \ref{mainTheorem} below, 
(\ref{conditonOnEndowment}) is equivalent
to the condition on random endowment in Hugonnier and Kramkov
\cite{HugonnierKramkov2004} (see equation (5)). As a result in order
to check that (\ref{conditonOnEndowment}) holds, one can use
alternative equivalent conditions presented in the statement of Lemma 1 in
\cite{HugonnierKramkov2004}, which in particular asserts that
$(x,0)\in\mathcal K$ for every $x>0$.
\end{Remark}

The preferences of an economic agent are modeled with a
\textit{utility stochastic field}
$U=U(t, \omega, x):[0,T]\times\Omega\times[0,\infty)\to\mathbb R\cup \{-\infty\}$.
As in \cite{Mostovyi2011}, we assume that $U$
satisfies the conditions below.
\begin{Assumption}
  \label{assumptionOnU}
  For every $(t, \omega)\in[0, T]\times\Omega$ the function $x\to
  U(t, \omega, x)$ is strictly concave, increasing, continuously
  differentiable on $(0,\infty)$ and satisfies the Inada conditions:
  \begin{equation}\label{Inada}
    \lim\limits_{x\downarrow 0}U'(t, \omega, x) =
    +\infty \quad \text{and} \quad \lim\limits_{x\to
      \infty}U'(t, \omega, x) = 0,
  \end{equation}
  where $U'$ denotes the partial derivative with respect to the third argument.
 At $x=0$ we have, by continuity, $U(t, \omega, 0) = \lim\limits_{x \downarrow 0}U(t, \omega, x)$, this value
 may be $-\infty$.
For every
  $x\geq 0$ the stochastic process $U\left( \cdot, \cdot, x \right)$ is
  optional.
\end{Assumption}

The agent can control investment and consumption. His goal is to do
this in a way that maximizes expected utility. The value function $u$ is
defined as:
\begin{equation}\label{primalProblem}
u(x, q) \set \sup\limits_{c\in\mathcal A(x, q)}\mathbb
E\left[\int_0^T U(t, \omega,c_t)d\kappa_t\right], \quad
(x, q)\in{\rm cl}\mathcal K.
\end{equation}
We use the convention
\begin{displaymath}
  \mathbb{E}\left[
    \int_0^TU(t,\omega, c_t)d\kappa_t \right] \set -\infty
  \quad \text{if} \quad \mathbb{E}\left[ \int_0^TU^{-}(t,\omega, c_t)d\kappa_t
  \right]= +\infty.
\end{displaymath}
Here and below, $W^{-}$ and $W^{+}$ denote the negative and the positive parts of a
stochastic field $W$, respectively.
 Also, we set $u(x,q) \set -\infty$ for $(x,q)\in\mathbb ({\rm
cl}\mathcal K)^c.$

We are primarily interested in the following
questions.
\begin{enumerate}[(i)]
\item Under what conditions on the market model and on the utility
stochastic field $U$ does the maximizer to the problem
(\ref{primalProblem}) exist for every $(x,q)\in\{u>-\infty\}$?
\item
What are the properties of the function $u?$
\item What is the corresponding dual problem?
\end{enumerate}
 We employ the duality techniques to answer these questions and define a convex conjugate stochastic field
\begin{equation}\nonumber
V(t,\omega,y) \set \sup\limits_{x>0}\left(
U(t,\omega,x)-xy\right),\quad (t,\omega,y)\in[0,T]\times\Omega\times[0,\infty).
\end{equation}
Observe that $-V$ satisfies Assumption \ref{assumptionOnU}.
In order to construct the feasible set of the dual problem we define the set
$\mathcal L$ as the relative interior of the  polar cone $-\mathcal
K$:
\begin{equation}\label{defOfL}
\mathcal L \set {\rm ri}\left\{(y,r)\in\mathbb R^{N+1}: xy + qr \geq
0{\rm~for~all~}(x,q)\in\mathcal K\right\}.
\end{equation}
It is proved that $\mathcal L$ is an open set in
$\mathbb R^{N+1}$ if and only if for every $q\neq 0$ the random
variable $qF_T$ is non-replicable, see Lemma 7 in
\cite{HugonnierKramkov2004} for the exact statement.

By ${\mathcal Z}$ we denote the set of c\'adl\'ag densities of
equivalent local martingale measures:
\begin{equation}\nonumber
{\mathcal Z} \set \left\{\left(\frac{d\mathbb Q_t}{d\mathbb P_t}\right)_{t\in[0,T]}:\quad \mathbb Q\in\mathcal M\right\},
\end{equation}
and for each $y \geq 0$ we define
\begin{equation}\label{oldY}
\begin{array}{c}
 \hspace{-15mm}\mathcal Y(y) \set {\rm cl}\left\{ Y: ~Y{\rm~is~c\acute{a}dl\acute ag~adapted~and~}\right. \\
 \hspace{35mm}\left.0\leq Y\leq yZ ~\left(d\kappa\times\mathbb
P\right){\rm
~a.e.~for~some~}Z\in{\mathcal Z}\right\},\\
\end{array}
\end{equation}
where the closure is taken in the topology of convergence in measure
$\left( d\kappa\times\mathbb P\right)$ on the space of optional
processes. Now we are ready to set the
domain of the dual problem:
\begin{equation}\label{defOfY}
\begin{array}{c}
\hspace{-30mm}
\mathcal Y(y, r) \set \left\{ Y\in\mathcal Y(y): \mathbb E\left[
\int_0^T c_tY_td\kappa_t\right]\leq xy + qr~~\right.\\
\hspace{17mm}
\left.{\rm for
~every~}(x,q)\in{\rm cl}\mathcal K ~{\rm and} ~c\in\mathcal A(x,q)\right\},~~(y,r)\in{\rm cl}\mathcal L,\\
\end{array}
\end{equation}
and to state the dual optimization problem itself:
\begin{equation}\label{dualProblem}
v(y, r) \set \inf\limits_{Y\in\mathcal Y(y, r)}\mathbb
E\left[\int_0^T V(t, \omega,Y_t)d\kappa_t\right],~~~~(y, r)\in{\rm
cl}\mathcal L,
\end{equation}
where we use the convention:
\begin{displaymath}
  \mathbb{E}\left[ \int_0^TV(t, \omega, Y_t
    )d\kappa_t \right] \set +\infty
  \quad \text{if} \quad \mathbb{E}\left[ \int_0^TV^{+}(t, \omega, Y_t )d\kappa_t
  \right] = +\infty.
\end{displaymath}
Also, we set $v(y,r) \set +\infty$ for all $(y,r)\in({\rm
cl}\mathcal L)^c.$

 The following
theorem constitutes the main result of this work.
\begin{Theorem}\label{mainTheorem}
Assume that (\ref{stochasticClock}), (\ref{MisNotEmpty}), Assumptions
\ref{assumptionOnEndowment} and \ref{assumptionOnU} hold as well as
\begin{equation}\label{cond1}
u(x, q) > -\infty~{for~every~}(x,q)\in\mathcal K{~and~} v(y, r)
<\infty~{for~every~}(y,r)\in\mathcal L.
\end{equation}
Then we have:
\begin{enumerate}[(i)]
\item
  The functions $u$ and $v$ are finite on $\mathcal K$
  and  $\mathcal L$, respectively, $u$ and $v$ satisfy biconjugacy relations:
\begin{equation}\label{conjugacy}
\begin{array}{rclc}
u(x,q) &=& \inf\limits_{(y,r)\in{\rm cl}\mathcal L}\left( v(y, r) +
xy + qr\right), &(x, q)\in{\rm cl}\mathcal K;
\\
v(y, r) &=& \sup\limits_{(x, q)\in{\rm cl}\mathcal K}\left(u(x, q) -
xy - qr\right), &(y, r)\in{\rm cl}\mathcal L.
\end{array}
\end{equation}

\item
The function $u$ is upper semi-continuous, $u(x,q)<\infty$ for every
$(x,q)\in{\rm cl} \mathcal K$. For every $(x,q)\in\{ u>-\infty\}$ there
exists a unique maximizer to the problem (\ref{primalProblem}).

 The function $v$ is lower
semi-continuous, $v(y,r)>-\infty$ for every $(y,r)\in{\rm cl} \mathcal L$.
For every $(y,r)\in\{v<\infty\}$ there exists a unique minimizer  to
the problem (\ref{dualProblem}).

\item
For every
$(x,q)\in \mathcal K$,
the subdifferential of $u$ at
$(x,q)$ is non-empty,  $(y, r)\in\partial u(x, q)$ if and
only if the following conditions hold:
\begin{equation}\label{17}
\hat Y_t(y,r) = U'\left(t, \omega, \hat
c_t(x,q)\right),~~~~t\in[0,T],
\end{equation}
where $\hat Y(y,r)$ and $\hat c(x,q)$ are optimizers to problems
(\ref{dualProblem}) and (\ref{primalProblem}), respectively;
\begin{equation}\label{18}
\mathbb E\left[ \int_0^T\hat Y_t(y,r)\hat c_t(x,q)d\kappa_t\right] =
xy + qr;
\end{equation}
\begin{equation}\label{7-082}
\left|v(y,r)\right|<\infty.
\end{equation}

\end{enumerate}

\end{Theorem}

Condition (\ref{cond1}) might be difficult to verify. The following lemma provides an equivalent
criterion in terms of the functions 
\begin{equation}\label{problemForW} w(x) \set u(x,0)=
\sup\limits_{c\in\mathcal A(x,0)}\mathbb E\left[
\int_0^TU\left(t,\omega,c_t \right)d\kappa_t
\right],\quad x>0,\end{equation}
and
\begin{equation}\label{problemForTildeW} \tilde w(y) \set
\inf\limits_{Y\in\mathcal Y(y)}\mathbb E\left[
\int_0^TV\left(t,\omega,Y_t \right)d\kappa_t
\right],\quad y>0.\end{equation}


\begin{Lemma}\label{Lemma2}
Let conditions (\ref{stochasticClock}) and (\ref{MisNotEmpty})  as well as Assumptions \ref{assumptionOnEndowment}
and
\ref{assumptionOnU} hold true. Then condition (\ref{cond1}) holds if
and only if 
\begin{equation}\label{cond2}
w(x) >-\infty{\rm~for~ every~} x>0{\rm~~ and~~} 
\tilde w(y)<\infty{\rm~ for~ every~} y>0.
\end{equation}
\end{Lemma}
Note that (\ref{cond2}) is precisely the condition that was used
by Mostovyi \cite{Mostovyi2011} in the statement of the main
theorem.

\section{Proofs}\label{proofOfMainThm}
We begin from a proposition that gives a useful characterization of
the primal and dual domains.
\begin{Proposition}\label{Proposition1}
Under the conditions (\ref{stochasticClock}), (\ref{MisNotEmpty}), and Assumption \ref{assumptionOnEndowment},
 the
families $\left( \mathcal A(x, q)\right)_{(x,q)\in{\rm
cl}\mathcal K}$ and $\left(\mathcal Y(y,r)\right)_{(y, r)\in{\rm
cl}\mathcal L}$ defined in (\ref{defOfA}) and (\ref{defOfY}) have
the following properties:
\begin{enumerate}[(i)]
\item For any $(x,q)\in\mathcal K,$ the set $\mathcal A(x, q)$
contains a strictly positive constant process. For every $(x,q)\in{\rm
cl}\mathcal K$ a nonnegative optional process $c$ belongs to
$\mathcal A(x, q)$ if and only if
\begin{equation}\label{26}
\mathbb E\left[ \int_0^Tc_tY_td\kappa_t\right]\leq xy + qr~~~~~{\rm
for~every~}(y, r)\in{\rm cl}\mathcal L~{\rm and}~Y\in\mathcal Y(y,r).
\end{equation}
\item For every $(y, r)\in\mathcal L$ the set $\mathcal Y(y,r)$
contains a strictly positive process. For every $(y,r)\in{\rm
cl}\mathcal L$ a nonnegative process $Y$ belongs to $\mathcal Y(y,
r)$ if and only if
\begin{equation}\label{27}
\mathbb E\left[ \int_0^Tc_tY_td\kappa_t\right]\leq xy + qr\quad{\rm
for~every~}(x, q)\in{\rm cl}\mathcal K~{\rm and}~c\in\mathcal A(x,q).
\end{equation}
\end{enumerate}
\end{Proposition}
The proof of Proposition \ref{Proposition1} is given via 
several lemmas.
As in \cite{HugonnierKramkov2004}, we define the set $\mathcal P$ to
be the set of points in the intersection of $\mathcal L$ and the
hyperplane $y\equiv 1,$ that is,
\begin{equation}\label{setP}
\mathcal P \set \left\{ p\in\mathbb R^{N}:(1,p)\in\mathcal
L\right\}.
\end{equation}
Note that under the Assumption \ref{assumptionOnEndowment} and (\ref{MisNotEmpty}), it
follows from  Lemma 1 in \cite{HugonnierKramkov2004} that the set
$\mathcal P$ is bounded.

Let $\mathcal M'$ be the set of equivalent local martingale measures
$\mathbb Q$, such that the process $X'$ (in the Assumption
\ref{assumptionOnEndowment}) is a uniformly integrable martingale under $\mathbb Q$.
According to Theorem 5.2 in Delbaen and Schachermayer
\cite{Delbaen-Schachermayer1997}, $\mathcal M'$ is a nonempty, convex
subset of $\mathcal M$, which is dense in $\mathcal M$ with respect
to the variation norm.

For every $p\in\mathcal P$ we denote
\begin{displaymath}
\mathcal M'(p) \set \left\{\mathbb
Q\in\mathcal M':~\mathbb {E_Q}\left[ F_T\right]= p \right\}.
\end{displaymath} It
follows from  Lemma 8 in \cite{HugonnierKramkov2004} that (under
condition  (\ref{MisNotEmpty}) and Assumption
\ref{assumptionOnEndowment}) $\mathcal M'(p)$ is non-empty for every
$p\in\mathcal P$ and
\begin{equation}\label{7-13}
\bigcup\limits_{p\in\mathcal P}\mathcal M'(p) = \mathcal
M',
\end{equation}
\begin{Lemma}\label{Lemma9}
Let the assumptions of Proposition \ref{Proposition1} hold true and $p\in\mathcal P$.
Then the c\'adl\'ag density process of any
$\mathbb Q\in\mathcal M'(p)$ belongs to $\mathcal Y(1, p)$.
\end{Lemma}
\begin{proof} Fix an arbitrary $(x,q)\in{\rm cl}\mathcal K,$ $c\in\mathcal
A(x,q),$ and $X\in\mathcal X(x, q)$ such that $X_T+qF_T \geq
\int_0^Tc_td\kappa_t\geq 0$. Notice that $X$ is a supermartingale under $\mathbb Q$. Therefore, 
taking expectation under $\mathbb
Q\in\mathcal M'(p)$ and using localization and integration by parts we get:
\begin{equation}\nonumber
x + qp \geq \mathbb {E_Q}\left[ \int_0^T c_td\kappa_t\right] =
\mathbb E\left[ \int_0^T\frac{d\mathbb Q_t}{d\mathbb
P_t}c_td\kappa_t\right].
\end{equation}
\end{proof}

\begin{Lemma}\label{Lemma10}
Let the assumptions of Proposition \ref{Proposition1} hold true.
Then for every $(x, q)\in{\rm cl}\mathcal K$, a nonnegative optional
process $c$ belongs to $\mathcal A(x,q)$ if and only if
\begin{equation}\label{32}
\mathbb {E_Q} \left[\int_0^T c_t d\kappa_t\right] \leq x + qp~~{\rm
for ~all ~}p\in\mathcal P~~{\rm and}~~\mathbb Q \in\mathcal M'(p).
\end{equation}
\end{Lemma}
\begin{proof} If $c\in\mathcal A(x, q)$ for $(x, q)\in{\rm cl}\mathcal K$ then
the validity of (\ref{32}) follows from Lemma \ref{Lemma9}.
Viceversa, let $c$ be a nonnegative optional process such that
(\ref{32}) holds. Denote
\begin{equation}\nonumber h \set
\int_0^Tc_td\kappa_t - qF_T,\quad M \set \max\limits_{1\leq i\leq
N}|q_i|.
\end{equation}
Then $h\geq -MX'_T$ and
\begin{equation}\nonumber
\begin{array}{rcl}
\alpha(h) \set \sup\limits_{\mathbb Q\in\mathcal M'}\mathbb
{E_Q}\left[ h\right] &=& \sup\limits_{p\in\mathcal
P}\sup\limits_{\mathbb
Q\in\mathcal M'(p)}\mathbb{E_Q}\left[ h\right]\\
&=&\sup\limits_{p\in\mathcal P}\sup\limits_{\mathbb Q\in\mathcal
M'(p)}\left(\mathbb {E_Q} \left[\int_0^T c_t
d\kappa_t\right]-qp\right)\leq x,\\
\end{array}
\end{equation}
where in the second equality we used (\ref{7-13}). Lemma 5 in
\cite{HugonnierKramkov2004} implies the existence of an acceptable
process $X$ such that $X_0 = \alpha(h)$ and $X_T\geq h.$ It follows
that
\begin{equation}\nonumber
X_T + qF_T \geq \int_0^Tc_td\kappa_t.
\end{equation}
Therefore $c\in\mathcal A(\alpha(h), q)\subseteq\mathcal A(x, q).$
\end{proof}

 \begin{proof}[Proof of Proposition \ref{Proposition1}]
 We prove the item $(i)$ first.
Fix $(x, q)\in\mathcal K$. Since $\mathcal K$ is an open set, there
exists $\delta >0$ such that $(x-\delta, q)\in\mathcal K.$ Take
$X\in\mathcal X(x-\delta,q)$ then $Z\set X+\delta\in\mathcal X(x, q)$. Consequently
\begin{equation}\nonumber
Z_T+qF_T \geq \delta \geq \int_0^T \left(\delta/A \right)d\kappa_t,
\end{equation}
where $A$ is the constant in (\ref{stochasticClock}). 
Therefore the process that takes the constant value $\delta/A$ belongs
to $\mathcal A(x,q)$.

Let $c$ be a nonnegative optional process such that (\ref{26})
holds. For every
$p\in\mathcal P$, it follows from Lemma \ref{Lemma9} that the c\'adl\'ag density process of
any $\mathbb Q\in\mathcal M'(p)$ is in $\mathcal Y(1, p)$.  Consequently, $c$ satisfies
(\ref{32}). It follows from Lemma \ref{Lemma10} that $c\in\mathcal
A(x,q).$ This concludes the proof of the item $(i)$.

To prove the assertion of the item $(ii),$  let us observe that
\begin{equation}\nonumber
a\mathcal Y(y,r) = \mathcal Y(ay, ar)\quad {\rm for~every~}a>0{\rm ~and~} (y, r)\in\mathcal L.
\end{equation}
Therefore it suffices to prove the existence of a strictly positive process for $(y,r)= (1,p)$, $p\in\mathcal
P$. Fix an arbitrary $p\in\mathcal P$. By Lemma 8 in \cite{HugonnierKramkov2004}, we deduce the existence of $\mathbb Q\in\mathcal
M'(p)$. By Lemma \ref{Lemma9}, the c\'adl\'ag density process
$\left( \frac{d\mathbb Q_t}{d\mathbb P_t}\right)_{t\in[0,T]}$ is in $\mathcal Y(1,p).$
Since $\mathbb Q$ is equivalent to $\mathbb P$,
$\left(\frac{d\mathbb Q_t}{d\mathbb P_t}\right)_{t\in[0,T]}$ is strictly positive $\mathbb P$ a.s.

Similarly, it suffices to consider $(y, r)\in{\rm cl}\mathcal L$ with $y=1$. For every $(1,p)\in {\rm cl}\mathcal L$, if $Y\in\mathcal Y(1,p)$,
condition (\ref{27}) follows from the definition of the set
$\mathcal Y(1, p).$ Conversely, let $Y$ be a nonnegative process
such that (\ref{27}) holds for $y=1$.  Then
\begin{equation}\nonumber
\mathbb E\left[\int_0^T c_tY_td\kappa_t \right] \leq 1\quad {\rm for~all~}c\in\mathcal A(1,0).
\end{equation}
Note that $\mathcal A(1, 0)$ is nonempty by 
Lemma 1 in \cite{HugonnierKramkov2004} (in view of  (\ref{MisNotEmpty}) and Assumption \ref{assumptionOnEndowment}).
Therefore, by Proposition 4.4 of \cite{Mostovyi2011} $Y\in\mathcal Y(1)$ (the
set $\mathcal Y(1)$ is defined in (\ref{oldY})) is such that (\ref{27})
holds, i.e. $Y\in\mathcal Y(1, p)$.

 \end{proof}

\begin{Lemma}\label{subConjugacy}
Under the conditions of Theorem \ref{mainTheorem}, for every $(x,q)\in{\rm
  cl}\mathcal K$ and
$(y,r)\in{\rm cl}\mathcal L$ we have
\begin{equation}\label{subConjugacyEqn}
u(x, q)\leq v(y, r) + xy + qr.
\end{equation}
As a result $u$  and $v$ are real-valued functions on $\mathcal K$ and  $\mathcal
L$, respectively, $u<\infty$ and $v>-\infty$ on $\mathbb R^{N+1}$.
\end{Lemma}
\begin{proof} Fix $(x, q)\in{\rm cl}\mathcal K$ and $(y, r)\in{\rm cl}\mathcal
L$. Take an arbitrary $c\in\mathcal A(x, q)$ and  $Y\in\mathcal Y(y,r)$.
It follows from the definition of the conjugate field $V$ that
\begin{equation}\nonumber
\begin{array}{rcl}
\mathbb E\left[ \int_0^T U(t, \omega, c_t)d\kappa_t\right] &\leq
&\mathbb E\left[ \int_0^T U(t, \omega, c_t)d\kappa_t\right] +xy +qr
- \mathbb E\left[\int_0^Tc_tY_td\kappa_t\right]\\
&\leq&\mathbb E\left[ \int_0^T V(t, \omega, Y_t)d\kappa_t\right] +xy
+ qr .\\
\end{array}
\end{equation}
This implies inequality (\ref{subConjugacyEqn}).
The remaining assertions of the lemma follow.
\end{proof}

Let $\mathbf L^0 = \mathbf L^0\left(d\kappa\times\mathbb P\right)$
be the vector space of optional process on the stochastic basis
$\left(\Omega, \mathcal F,(\mathcal F_t)_{t\in[0,T]},\mathbb P\right).$

\begin{Lemma}\label{upperSemiContinuity} Let the conditions of Theorem \ref{mainTheorem} hold true. 
Then the function $u$ is upper
semi-continuous. For every $(x,q)\in\{ u>-\infty\}$ there exists a
unique maximizer to the problem (\ref{primalProblem}).
Likewise, the function $v$ is lower
semi-continuous. For every $(y,r)\in\{v<\infty\}$ there exists a
unique minimizer  to the problem (\ref{dualProblem}).
\end{Lemma}
\begin{proof}
Let $(y^n,r^n)_{n\geq 1}$ be a sequence in $\mathcal L$ converging to a point $(y,r)\in\{ v<\infty\}$.
Let us fix $Y^n\in\mathcal Y(y^n, r^n)$ such that
\begin{displaymath}
\mathbb E \left[ \int_0^T V\left(t, \omega, Y^n_t\right)d\kappa_t\right] \leq v(y^n, r^n) +\frac{1}{n},\quad n\geq 1.
\end{displaymath}

By Lemma A1.1 in \cite{DS}, there exists a sequence of convex combinations $\tilde Y^n\in {\rm conv}(Y^n,Y^{n+1},\dots)$, $n\geq 1$,
and an element $\hat Y\in\mathbf L^0$, such that $\left( \tilde Y^n\right)_{n\geq 1}$ converges to $\hat Y$ $\left( d\kappa\times \mathbb P\right)$
a.e.

For every $(x,q)\in {\rm cl}\mathcal K$ and $c\in\mathcal A(x,q)$, using Fatou's lemma we get:
\begin{displaymath}
\mathbb E\left[ \int_0^T\hat Y_tc_td\kappa_t\right]\leq
\liminf\limits_{n\to\infty}\mathbb E\left[\int_0^Tc_t\tilde Y^n_td\kappa_t \right]
\leq xy +qr.
\end{displaymath}
Consequently, using Proposition  \ref{Proposition1} we deduce that $\hat Y\in\mathcal Y(y,r)$.

By Lemma \ref{Lemma2} the functions $w$ and $\tilde w$ (defined in (\ref{problemForW}) and (\ref{problemForTildeW}) respectively) satisfy the assumptions of Theorem 3.2 in \cite{Mostovyi2011}.
Let $\bar y \set \sup\limits_{n\geq 1}|y^n|$, then $\left(Y^n\right)_{n\geq 1}\subseteq \mathcal Y(\bar y)$. Therefore, from Lemma 3.5 in \cite{Mostovyi2011}
we deduce that the sequence $\left( V^{-}\left( t, \omega, Y^n_t\right)\right)_{n\geq 1}$
is uniformly integrable. Consequently, using convexity of $V$, we obtain
\begin{equation}\label{lscMainFormula}
\begin{array}{rcl}
v(y,r) &\leq& \mathbb E\left[ \int_0^TV(t, \omega, \hat Y_t)d\kappa_t\right] \\
&\leq&
\liminf\limits_{n\to\infty}\mathbb E\left[ \int_0^TV(t, \omega, Y^n_t)d\kappa_t\right]
=
\liminf\limits_{n\to\infty}v\left(y^n,r^n\right),\\
\end{array}
\end{equation}
which implies lower semi-continuity of $v$. Moreover, by Lemma \ref{subConjugacy} $v>-\infty$ everywhere in its domain. As a result for every $(y,r)\in\{ v<\infty\}$, taking $(y^n,r^n) = (y,r)$, $n\geq 1$, we deduce from (\ref{lscMainFormula}) the existence of a minimizer to the dual problem (\ref{dualProblem}), whose uniqueness follows from the {strict} convexity of $V$. The proof of the corresponding assertions for the function $u$ is similar.
\end{proof}

For each $(y,r)\in\mathcal L$  define the following sets:
\begin{equation}\label{A}
A(y,r) \set \left\{ (x,q)\in\mathcal K:~xy + qr \leq 1\right\},
\end{equation}
\begin{equation}\label{tildeC}
\tilde {\mathcal C}(y, r) \set \bigcup\limits_{(x, q)\in A(y,
r)}\mathcal A(x, q),
\end{equation}
\begin{equation}\label{C}
{\mathcal C}(y, r) \set {\rm cl} \tilde {\mathcal C}(y,r),
\end{equation}
where closure is taken in the topology of convergence in measure
$\left(d\kappa\times \mathbb P\right).$ For the proof of Theorem
\ref{mainTheorem} we need the following lemma.
\begin{Lemma}\label{Lemma11}
Let $(y, r)\in\mathcal L,$ $\tilde {\mathcal C}(y,r)$ and ${\mathcal
C}(y,r)$ be given by (\ref{tildeC}) and (\ref{C}) respectively. Then, under the
conditions of Theorem \ref{mainTheorem}, we have
\begin{equation}\nonumber
\sup\limits_{g\in\tilde {\mathcal C}(y,r)}\mathbb E\left[ \int_0^T
U(t, \omega, xc_t)d\kappa_t\right] = \sup\limits_{g\in {\mathcal
C}(y,r)}\mathbb E\left[ \int_0^T U(t, \omega,
xc_t)d\kappa_t\right],\quad x>0.
\end{equation}
\end{Lemma}
\begin{proof}
For each $x>0$ let us define
\begin{equation}\nonumber
\phi(x) \set \sup\limits_{g\in\tilde {\mathcal C}(y,r)}\mathbb
E\left[ \int_0^T U(t, \omega, xc_t)d\kappa_t\right],~~ \psi(x) \set
\sup\limits_{g\in {\mathcal C}(y,r)}\mathbb E\left[ \int_0^T U(t,
\omega, xc_t)d\kappa_t\right].
\end{equation}
Due to concavity of $U$, both $\phi$ and $\psi$ are concave,  and $\phi\leq\psi$. If $\phi(x) = \infty$ for some
$x >0$ then, due to concavity, $\phi$ is infinite for all $x>0$. In this case the assertion of the theorem is trivial.
Also, it follows from Lemma \ref{subConjugacy} that $\phi(x) > -\infty$ for
every $x>0$.
Therefore, without loss of generality for the remainder of this proof we will assume that
$\phi$ is finite.

Fix $x>0$ and $g\in \mathcal C(y, r)$. Let $\left( g^n\right)_{n\geq
1}$ be a sequence in $\tilde {\mathcal C}(y,r)$ that converges  $\left( d\kappa\times\mathbb P\right) $ almost everywhere to
$g$. It
follows from Lemma \ref{subConjugacy} that for every $\delta >0$ there
exists $c\in\tilde {\mathcal C}(y,r)$ such that
\begin{equation}\nonumber
 \mathbb E\left[ \int_0^T U(t, \omega, \delta c_t)d\kappa_t\right] > -\infty.
\end{equation}
 Therefore we have:
 \begin{equation}\nonumber
 \begin{array}{rcl}
 \mathbb E\left[ \int_0^TU(t, \omega,xg_t)d\kappa_t\right] &\leq &
\mathbb E\left[ \int_0^TU(t, \omega,xg_t + \delta c_t)d\kappa_t\right]\\
&\leq&\liminf\limits_{n\to\infty}\mathbb E\left[ \int_0^TU(t, \omega,xg^n_t + \delta c_t)d\kappa_t\right]
\\
&\leq&\phi(x + \delta),\\
\end{array}
 \end{equation}
where the first inequality is valid because $U$ is increasing,
 the second one
follows from Fatou's lemma, and the third one comes from the fact
that $\tilde {\mathcal C}(y,r)$ is convex.
 Since $\phi$ is concave, it is continuous.
As a result
\begin{equation}\nonumber
\psi(x) = \sup\limits_{g\in {\mathcal C}(y,r)} \mathbb
E\left[\int_0^T U(t,\omega,xg_t)d\kappa_t\right] \leq
\lim\limits_{\delta\to 0}\phi(x + \delta) = \phi(x).
\end{equation}
\end{proof}

\begin{proof}[Proof of the Theorem \ref{mainTheorem}] \textit{(i)} Concavity of the
function $u$ follows from strict concavity  of $U$.
Fix $(y, r)\in\mathcal L$. Applying Lemma \ref{Lemma11} and the definition of the set $\tilde
{\mathcal C}(y,r),$  we get for each $z>0$:
\begin{equation}\label{barUisFinite}
\begin{array}{c}
\bar u(z) \set \sup\limits_{c\in\mathcal C(y,r)}\mathbb E\left[
\int_0^T U(t, \omega, zc_t)d\kappa_t\right] = \\ = \sup\limits_{c\in
\tilde{\mathcal C}(y,r)}\mathbb E\left[ \int_0^T U(t, \omega,
zc_t)d\kappa_t\right] =\sup\limits_{(x, q)\in zA(y,r)}u(x,q)
>-\infty.
\end{array}
\end{equation}
It follows from Proposition \ref{Proposition1} that
\begin{equation}\nonumber
Y\in\mathcal Y(y,r) ~~~\Leftrightarrow ~~~\mathbb E\left[ \int_0^T
c_tY_td\kappa_t\right]\leq 1~~~{\rm for~all~}c\in\mathcal C(y,r).
\end{equation}
We obtain that the sets ${\mathcal C}(y,r)$ and $\mathcal Y(y,r)$
satisfy the assumption of Theorem 3.2 in Mostovyi
\cite{Mostovyi2011}.
Consequently, since $v(y,r) < \infty$ for all $(y,r)\in\mathcal L,$ using
(\ref{barUisFinite}) we get:
\begin{equation}\nonumber
v(y,r) = \sup\limits_{z>0}\left( \bar u(z) - z\right) =
\sup\limits_{(x,q)\in\mathcal K}\left(u(x,q) - xy - qr\right).
\end{equation}
It follows from Lemma \ref{upperSemiContinuity} that  $-u$ and $v$
are proper closed convex functions. Therefore the latter equality implies
the biconjugacy relations (\ref{conjugacy}) (see Rockafellar \cite{Rok},
Section 12).

\textit{(ii)} The assertions of item $(ii)$ follow from Lemma \ref{upperSemiContinuity}.

\textit{(iii)} Conjugacy relations  (\ref{conjugacy}) imply (by Theorem 23.4 and
Corollary 23.5.1 in Rockafellar \cite{Rok}) that for every $(x, q)\in\mathcal
K$ we have $\partial
u(x,q) \subseteq cl \mathcal L$.

Let $(x, q)\in\mathcal K$ and $(y,r)\in {\rm cl}\mathcal L$ be such that 
(\ref{17}), (\ref{18}), and (\ref{7-082}) hold, where $\hat c(x,q)$ and $\hat Y(y,r)$ the optimizers to (\ref{primalProblem}) and 
(\ref{dualProblem}) respectively. The existence of such optimizers follows from Lemma \ref{upperSemiContinuity}.
Then
using conjugacy of $U$ and $V$ we obtain:
\begin{displaymath}
\begin{array}{rcl}
 0&=& \mathbb E\left[ \int_0^T \left(V(t,\omega,\hat Y_t(y,r)) -
U(t, \omega, \hat c_t(x,q)) + \hat c_t(x,q)\hat
Y_t(y,r)\right)d\kappa_t\right] \\
&=& v(y,r) - u(x,q) + xy + qr.\\
\end{array}
\end{displaymath}
By Theorem 23.5 in
\cite{Rok}, the biconjugacy relations (\ref{conjugacy}) imply 
that $(y,r)\in\partial u(x,q)$.

Conversely, fix $(x,q)\in\mathcal K$,
and let $(y,r)\in\partial u(x,q)$.
By Lemma \ref{upperSemiContinuity} $-u$ and $v$
are closed convex functions. We have also proved in item \textit{(i)} that they satisfy biconjugacy
relations (\ref{conjugacy}). Consequently,
\begin{equation}\label{7-08}
-u(x,q) + v(y,r) + xy + qr \leq 0.
\end{equation}
Now using Lemma \ref{subConjugacy}, we deduce (\ref{7-082}). In turn, by Lemma \ref{upperSemiContinuity} this implies
that there exists $\hat Y(y,r)$, a unique minimizer to the problem
(\ref{dualProblem}). As $u(x,y)>-\infty$, by Lemma
\ref{upperSemiContinuity} we deduce that there exists $\hat c(x,q)$,
a unique maximizer to the problem (\ref{primalProblem}). Using
Proposition \ref{Proposition1} we obtain from (\ref{7-08}):
\begin{equation}\nonumber
\begin{array}{c}
\mathbb E\left[ \left| \int_0^T\left( V\left(t, \omega,\hat Y_t(y,r)
\right) + \hat c_t(x,q)\hat Y_t(y,r) - U\left(t, \omega,\hat
c_t(x,q)\right)\right)d\kappa_t\right|\right] \\
= \mathbb E\left[ \int_0^T\left( V\left(t, \omega,\hat Y_t(y,r)
\right) + \hat c_t(x,q)\hat Y_t(y,r) - U\left(t, \omega,\hat
c_t(x,q)\right)\right)d\kappa_t\right] \\
\leq v(y,r) + xy + qr - u(x,q)\leq 0,\\
\end{array}
\end{equation}
which gives (\ref{17}) and (\ref{18}).

\end{proof}

\begin{proof}[Proof of Lemma \ref{Lemma2}]
 Assume that (\ref{cond2}) holds. Fix
$(x,q)\in\mathcal K.$ It follows from Assumption \ref{assumptionOnEndowment}
and Lemma 1 in \cite{HugonnierKramkov2004} that
$(x,0)\in K$ for each $x>0.$ Since $\mathcal K$ is an open convex cone, there
exists a point $(x_1,q_1)\in\mathcal K,$ such that
$$(x,q) = \lambda (x_1,q_1) + (1-\lambda)(x_2,0)$$ for some $\lambda
\in(0,1)$ and $x_2>0$. 
Take $c\in\mathcal A(x_2,0),$ such that

\begin{equation}\label{7-09-1}
\mathbb E\left[\int_0^TU(t,\omega,(1-\lambda)c_t)d\kappa_t \right] >
-\infty.
\end{equation}
 Note that such a process $c$ exists by
assumption (\ref{cond2}). Fix $g\in\mathcal A(x_1,q_1)$. Then we have $$\lambda g +
(1-\lambda)c\in\mathcal A(x,q).$$ Since $U$ is
increasing, we obtain from
(\ref{7-09-1}):
\begin{equation}\nonumber
\begin{array}{c}
u(x,q) \geq \mathbb E\left[\int_0^TU(t,\omega,\lambda g_t +
(1-\lambda)c_t)d\kappa_t \right] \geq\\
\geq \mathbb E\left[\int_0^TU(t,\omega,(1-\lambda)c_t)d\kappa_t
\right] >
-\infty.\\
\end{array}
\end{equation}

In order to prove that $v$ is finite on $\mathcal L$, define
the set
\begin{equation}\nonumber
\mathcal E \set \left\{ (y,r)\in{\rm cl}\mathcal L:~~
v(y,r)<\infty\right\}.
\end{equation}
First, we show that $\mathcal E$ is nonempty and establish some properties of
$\mathcal E$.
Let
\begin{displaymath}
\begin{array}{rcl}
\mathcal B &\set& \left\{ (y,r) \in\mathcal L:~~y\leq 1\right\},\\
\tilde{\mathcal  D} &\set& \bigcup\limits_{(y,r) \in \mathcal B}\mathcal Y(y,r),\\
\mathcal D &\set& {\rm cl}\tilde{\mathcal D},\\\end{array}
\end{displaymath}
where the closure is taken in measure $\left(d\kappa\times \mathbb P \right)$.
It follows from Proposition  \ref{Proposition1} and Fatou's lemma that

\begin{equation}\nonumber
c\in\mathcal A(1, 0) \quad\Leftrightarrow \quad\mathbb E\left[ \int_0^T
c_tY_td\kappa_t\right]\leq 1\quad{\rm for~every~}Y\in\mathcal D.
\end{equation}
Therefore the sets $\mathcal A(1,0)$ and $\mathcal D$ satisfy
the assumptions of Theorem 3.2 in \cite{Mostovyi2011}, which in particular
asserts that for every $x>0$ there exsits $\hat c(x)$, a unique maximizer  to
(\ref{problemForW}). Thus, for every $x>0$, we define
\begin{equation}\nonumber
 Y_t(x) \set U'\left(t,\omega, \hat
c_t(x)\right),\quad t\in[0,T].
\end{equation}
It follows from the same theorem that $w$ is a continuously differentiable
function that satisfies the Inada conditions and
\begin{displaymath}
Y(x)\in w'(x)\mathcal D.
\end{displaymath}
Using Proposition  \ref{Proposition1} and Fatou's lemma we can show that there
exists $(y,r)\in w'(x){\rm cl}\mathcal B$ .
Therefore, $\mathcal E\neq
\emptyset$. Moreover, since $w$ satisfies the Inada conditions, we deduce that
the closure of $\mathcal E$ contains origin. One can also see that the set $\mathcal
E$ is convex and
\begin{displaymath}
\mathcal E \supseteq \bigcup\limits_{\lambda \geq 1}\lambda \mathcal E.
\end{displaymath}

Second, we prove that $ \mathcal L \subseteq \mathcal E$. Fix an arbitrary
$(y,r)\in\mathcal L$ and let $\delta >0$ be such that $B_{\delta}(y,r)\subset
\mathcal L$, where $B_{\delta}(y, r)$ denotes the ball in $\mathbb R^{N+1}$ of radius $\delta$
centered at $(y,r)$. Since origin is in the closure of $\mathcal E$, there
exists $(\tilde y_2, \tilde r_2)\in\mathcal E\cap B_{\delta/2}(0)$. Let
\begin{displaymath}
(\tilde y_1, \tilde r_1)\set (y-\tilde y_2, r - \tilde r_2).
\end{displaymath}
Then $(\tilde y_1, \tilde r_1)\in B_{\delta/2}(y,r)$. Therefore, there exists
$\lambda \in (0,1)$ such that
\begin{displaymath}
(y_1,r_1)\set\frac{1}{\lambda}(\tilde y_1, \tilde r_1)\in B_{\delta}(y,r).
\end{displaymath}
Set $(y_2, r_2)\set \frac{1}{1-\lambda}(\tilde y_2, \tilde r_2)$, then
\begin{displaymath}
(y,r) = \lambda (y_1,r_1) + (1-\lambda)(y_2,r_2).
\end{displaymath}
Fix a process $Y'\in\mathcal Y(y_1,r_1)$. By
construction of the set $\mathcal E$ there exists a process
$Y''\in\mathcal Y(y_2,r_2),$ such that
\begin{equation}\nonumber
\mathbb E\left[
\int_0^TV(t,\omega,(1-\lambda)Y''_t)d\kappa_t\right]<\infty.
\end{equation}
Since $V$ is decreasing and $\left(\lambda Y' + (1-\lambda) Y''\right)\in\mathcal Y(y,r)$, we deduce
\begin{equation}\nonumber
\begin{array}{c}
v(y,r) \leq \mathbb E\left[ \int_0^TV(t,\omega,\lambda
Y'_t+(1-\lambda)Y''_t)d\kappa_t\right] \\
\\ \leq \mathbb E\left[
\int_0^TV(t,\omega,(1-\lambda)Y''_t)d\kappa_t\right]<\infty. \\
\end{array}
\end{equation}

Conversely, if (\ref{cond1}) holds then for every $p\in\mathcal P$, since
$\mathcal Y(y,yp)$ is a subset of $\mathcal Y(y)$, we have
\begin{displaymath}
\tilde w(y) \leq v(y, yp) < \infty,\quad y>0.
\end{displaymath}
The other assertion of (\ref{cond2}) follows trivially.
 \end{proof}

\section*{Acknowledgements}
I would like to thank Giovanni Leoni, Dmitry Kramkov, and Pietro Siorpaes for
the discussions on the topics of the paper.

\bibliographystyle{plainnat} \bibliography{finance}

\begin{thebibliography}{20}
\providecommand{\natexlab}[1]{#1}
\providecommand{\url}[1]{\texttt{#1}}
\expandafter\ifx\csname urlstyle\endcsname\relax
  \providecommand{\doi}[1]{doi: #1}\else
  \providecommand{\doi}{doi: \begingroup \urlstyle{rm}\Url}\fi

\bibitem[Cuoco(1997)]{Cuoco1997}
D.~Cuoco.
\newblock Optimal consumption and equilibrium prices with portfolio constraints
  and stochastic income.
\newblock \emph{J. Econom. Theory}, 72:\penalty0 33--73, 1997.

\bibitem[Cvitani\'c et~al.(2001)Cvitani\'c, Schachermayer, and
  Wang]{CvitanicSchachermayerWang}
J.~Cvitani\'c, W.~Schachermayer, and H.~Wang.
\newblock Utility maximization in incomplete markets with random endowment.
\newblock \emph{Finance Stoch.}, 5:\penalty0 259--272, 2001.

\bibitem[Delbaen and Schachermayer(1994)]{DS}
F.~Delbaen and W.~Schachermayer.
\newblock A general version of the fundamental theorem of asset pricing.
\newblock \emph{Math. Ann.}, 300:\penalty0 463--520, 1994.

\bibitem[Delbaen and Schachermayer(1997)]{Delbaen-Schachermayer1997}
F.~Delbaen and W.~Schachermayer.
\newblock The {B}anach space of workable contingent claims in arbitrage theory.
\newblock \emph{Ann. Inst. H. Pincar\'{e} Statist. Probab.}, 33:\penalty0
  113--144, 1997.

\bibitem[Delbaen and Schachermayer(1998)]{Delbaen-Schachermayer1998}
F.~Delbaen and W.~Schachermayer.
\newblock The fundamental theorem of asset pricing for unbounded stochastic
  processes.
\newblock \emph{Math. Ann.}, 312:\penalty0 215--250, 1998.

\bibitem[Duffie and Zariphopoulou(1993)]{DuffieZarip1993}
D.~Duffie and T.~Zariphopoulou.
\newblock Optimal investment with undiversifiable income risk.
\newblock \emph{Math. Finance}, 3:\penalty0 135--148, 1993.

\bibitem[Duffie et~al.(1997)Duffie, Fleming, Soner, and
  Zariphopoulou]{DuffieFlemingSonerZarip1997}
D.~Duffie, W.~Fleming, M.~Soner, and T.~Zariphopoulou.
\newblock Heding in incomplete markets with {HARA} utility.
\newblock \emph{J. Econ. Dymaics and Control}, 21:\penalty0 753--782, 1997.

\bibitem[He and Pearson(1991{\natexlab{a}})]{HePearson1}
H.~He and N.~D. Pearson.
\newblock Consumption and portfolio policies with incomplete markets and
  short-sale constraints: the finite-dimensional case.
\newblock \emph{Math. Finance}, 1:\penalty0 1--10, 1991{\natexlab{a}}.

\bibitem[He and Pearson(1991{\natexlab{b}})]{HePearson2}
H.~He and N.~D. Pearson.
\newblock Consumption and portfolio policies with incomplete markets and
  short-sale constraints: the infinite-dimensional case.
\newblock \emph{J. Econom. Theory}, 54:\penalty0 259--304, 1991{\natexlab{b}}.

\bibitem[Hugonnier and Kramkov(2004)]{HugonnierKramkov2004}
J.~Hugonnier and D.~Kramkov.
\newblock Optimal investment with random endowment in incomplete markets.
\newblock \emph{Ann. Appl. Probab.}, 14:\penalty0 845--864, 2004.

\bibitem[Karatzas and Kardaras(2007)]{Karatzas-Kardaras}
I.~Karatzas and K.~Kardaras.
\newblock The num\'{e}raire portfolio in semimartingale financial models.
\newblock \emph{Finance Stoch.}, 11:\penalty0 447--493, 2007.

\bibitem[Karatzas and Shreve(1998)]{KSmmf}
I.~Karatzas and S.~Shreve.
\newblock \emph{Methods of Mathematical Finance}.
\newblock Springer, 1998.

\bibitem[Karatzas and \v{Z}itkovi\'{c}(2003)]{Karatzas-Zitkovic-2003}
I.~Karatzas and G.~\v{Z}itkovi\'{c}.
\newblock Optimal consumption from investment and random endowment in
  incomplete semi-martingale markets.
\newblock \emph{Ann. Probab.}, 31:\penalty0 1821--1858, 2003.

\bibitem[Karatzas et~al.(1991)Karatzas, Lehoczky, Shreve, and Xu]{KLSX}
I.~Karatzas, J.~P. Lehoczky, S.~E. Shreve, and G.~L. Xu.
\newblock Martingale and duality methods for utility maximization in an
  incomplete market.
\newblock \emph{SIAM J. Control Optim.}, 29:\penalty0 702--730, 1991.

\bibitem[Kardaras and \v{Z}itkovi\'{c}(2011)]{KardarasZitkovic2011}
K.~Kardaras and G.~\v{Z}itkovi\'{c}.
\newblock Stability of the utility maximization problem with random endowment
  in incomplete markets.
\newblock \emph{Math. Finance}, 21:\penalty0 313--333, 2011.

\bibitem[Kramkov and Schachermayer(1999)]{KS}
D.~Kramkov and W.~Schachermayer.
\newblock The asymptotic elasticity of utility functions and optimal investment
  in incomplete markets.
\newblock \emph{Ann. Appl. Probab.}, 9:\penalty0 904--950, 1999.

\bibitem[Kramkov and Schachermayer(2003)]{KS2003}
D.~Kramkov and W.~Schachermayer.
\newblock Necessary and sufficient conditions in the problem of optimal
  investment in incomplete markets.
\newblock \emph{Ann. Appl. Probab.}, 13:\penalty0 1504--1516, 2003.

\bibitem[Mostovyi(2011)]{Mostovyi2011}
O.~Mostovyi.
\newblock Necessary and sufficient conditions in the problem of optimal
  investment with intermediate consumption.
\newblock \emph{arXiv:1107.5852v1 [q-fin.PM]}, 2011.

\bibitem[Rockafellar(1970)]{Rok}
R.~T. Rockafellar.
\newblock \emph{Convex Analysis}.
\newblock Princeton Univ. Press., 1970.

\bibitem[\v{Z}itkovi\'{c}(2005)]{Zitkovic}
G.~\v{Z}itkovi\'{c}.
\newblock Utility maximization with a stochastic clock and an unbounded random
  endowment.
\newblock \emph{Ann. Appl. Probab.}, 15:\penalty0 748--777, 2005.

\end{thebibliography}

\end{document}